\documentclass[11pt]{article}

\usepackage{times}
\usepackage{amsthm}
\usepackage[cm]{fullpage}
\usepackage{amsfonts}
\usepackage{amssymb}
\usepackage{graphicx}
\usepackage{algpseudocode}
\usepackage{amsmath}
\usepackage{constants}
\usepackage{graphicx}
\usepackage[ruled]{algorithm2e}

\newcommand{\var}{{\bf Var}}
\newcommand{\E}{{\bf E}}
\newcommand{\abs}[1]{| #1 |}
\renewcommand{\Pr}{{\bf Pr}}

\newcommand{\match}{{\small{\sc DMR}}}

\newcommand{\AND}{{\small{\sc AND}}}
\newcommand{\DISJ}{{\small{\sc DISJ}}}

\renewcommand{\P}{\mathcal{P}}

\providecommand{\abs}[1]{\left|#1\right|}

\theoremstyle{plain}
\newtheorem{theorem}{Theorem}[section]%[chapter]
\newtheorem{lemma}[theorem]{Lemma}

\newtheorem{definition}[theorem]{Definition}

\theoremstyle{definition}

\begin{document}

\title{Communication complexity of approximate maximum matching in the message-passing model
%\thanks{This paper is a revised and extended version of a paper by the same authors that appeared in the Proceedings of the 32nd Symposium on Theoretical Computer Science (STACS), Munich, Germany, March 4-7, 2015.}
%}
}


\author{Zengfeng Huang\thanks{University of New South Wales, Sydney, Australia. Email: {zengfeng.huang@unsw.edu.au}}     
	\and
        Bozidar Radunovic\thanks{Microsoft Research, Cambridge, United Kingdom.
        	             Email: {bozidar@microsoft.com}} 		 
    \and
				Milan Vojnovic\thanks{Department of Statistics,
					             London School of Economics (LSE), London, United Kingdom.
					              Email: {m.vojnovic@lse.ac.uk}}			   
	\and
				Qin Zhang\thanks{Computer Science Department, Indiana University, Bloomington, USA.
					              Email: {qzhangcs@indiana.com}}
}


%\institute{Zengfeng Huang \at
%              School of Computer Science and Engineering, The University of New South Wales, Australia\\
%              \email{huangzf@cse.ust.hk}           %  \\
%           Bozidar Radunovic \at
%              Microsoft Research, Cambridge, United Kingdom\\
%              \email{bozidar@microsoft.com}
%					 \and					
%	         Milan Vojnovic \at
%						  Department of Statistics,
%              London School of Economics (LSE), London, United Kingdom\\
%              \email{m.vojnovic@lse.ac.uk}
%					\and
%					 Qin Zhang \at
%              Computer Science Department, Indiana University, Bloomington, USA \\
%              \email{qzhangcs@indiana.com}
%							}

\date{}

\maketitle

\begin{abstract}
	We consider the communication complexity of finding an approximate maximum matching in a graph in a multi-party message-passing communication model. The maximum matching problem is one of the most fundamental graph combinatorial problems, with a variety of applications.
	
	The input to the problem is a graph $G$ that has $n$ vertices and the set of edges partitioned over $k$ sites, and an approximation ratio parameter $\alpha$. The output is required to be a matching in $G$ that has to be reported by one of the sites, whose size is at least factor $\alpha$ of the size of a maximum matching in $G$. 
	
	We show that the communication complexity of this problem is $\Omega(\alpha^2 k n)$ information bits. This bound is shown to be tight up to a $\log n$ factor, by constructing an algorithm, establishing its correctness, and an upper bound on the communication cost. The lower bound also applies to other graph combinatorial problems in the message-passing communication model, including max-flow and graph sparsification. 
%\input{abstract}
%Include keywords, PACS and mathematical subject classification numbers as needed.
%\keywords{Approximate maximum matching \and Multi-party communication complexity \and Message passing}
\end{abstract}

\section{Introduction}\label{sec:intro}
Complex and massive volume data processing requires to scale out to parallel and distributed computation platforms. Scalable distributed computation algorithms are needed that make efficient use of scarce system resources such as communication bandwidth between compute nodes in order to avoid the communication network becoming a bottleneck. A particular interest has been devoted to studying scalable computation methods for graph data, which arises in a variety of applications including online services, online social networks, biological, and economic systems.

In this paper, we consider the distributed computation problem of finding an approximate maximum matching in an input graph whose edges are partitioned over different compute nodes (we refer to as sites). Several performance measures are of interest including the communication complexity in terms of the number of bits or messages, the time complexity in terms of the number of rounds, and the storage complexity in terms of the number of bits. In this paper we focus on the communication complexity. Our main result is a tight lower bound on the communication complexity for approximate maximum matching. 

We assume a multi-party message-passing communication model~\cite{BEOPV13,PVZ12}, we refer to as \emph{message-passing model}, which is defined as follows. The message-passing model consists of $k\ge 2$ sites $p^1$, $p^2$, $\ldots$, $p^k$. The input is partitioned across $k$ sites, with sites $p^1$, $p^2$, $\ldots$, $p^k$ holding pieces of input data $x^1$, $x^2$, $\ldots$, $x^k$, respectively. The goal is to design a communication \emph{protocol} for the sites to jointly compute the value of a given function $f:{\mathcal X}^k \rightarrow {\mathcal Y}$ at point $(x^1,x^2, \ldots, x^k)$. The sites are allowed to have point-to-point communications between each other. At the end of the computation, at least one site should return the answer. The goal is to find a protocol that minimizes the total communication cost between the sites. 

For technical convenience, we introduce another special party called the {\em coordinator}. The coordinator does not have any input. We require that all sites can only talk with the coordinator, and at the end of the computation, the coordinator should output the answer. We call this model {\em the coordinator model}. See Figure~\ref{fig:cor} for an illustration. Note that we have essentially  replaced the clique communication topology with a star topology, which increases the total communication cost only by a factor of $2$ and thus, it does not affect the order of the asymptotic communication complexity.

The \emph{edge partition} of an input graph $G = (V,E)$ over $k$ sites is defined by a partition of the set of edges $E$ in $k$ disjoint sets $E^1$, $E^2$, $\ldots$, $E^k$, and assigning each set of edges $E^i$ to site $p^i$. For bipartite graphs with a set of left vertices and a set of right vertices, we define an alternative way of an edge partition, referred to as the \emph{left vertex partition}, as follows: the set of left vertices are partitioned in $k$ disjoints parts, and all the edges incident to one part is assigned to a unique site. Note that left vertex partition is more restrictive, in the sense that any left vertex partition is an instance of an edge partition. Thus, lower bounds holds in this model are stronger as designing algorithms might be easier in this restrictive setting. Our lower bound is proved for left vertex partition model, while our upper bound holds for an arbitrary edge partition of any graph.

\begin{figure}[t]
\begin{center}
\includegraphics[width=0.45\textwidth]{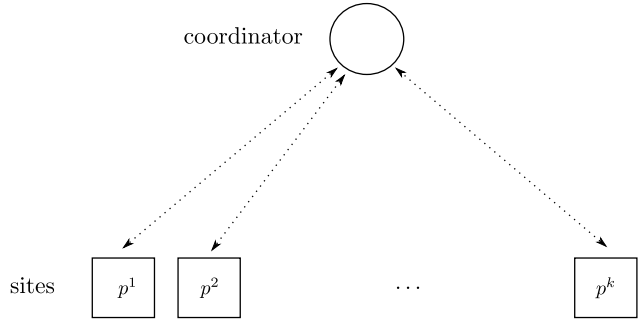}
\end{center}
\caption{Coordinator model.}
\label{fig:cor}
\end{figure}

\subsection{Summary of results} \label{sec:results}

We study the approximate maximum matching problem in the message-passing model which we refer to as Distributed Matching Reporting (\match) that is defined as follows: given as input is a graph $G = (V,E)$ with $\abs{V} = n$ vertices and a parameter $0 < \alpha \le 1$; the set of edges $E$ is arbitrarily partitioned into $k\ge 2$ subsets $ E^1, E^2, \cdots , E^k$ such that $E^i$ is assigned to site $p^i$; the coordinator is required to report an $\alpha$-approximation of the maximum matching in graph $G$. 

In this paper, we show the following main theorem.

\begin{theorem}
\label{thm:main}
For every $0<\alpha \le 1$ and the number of sites $1< k\le n$, any $\alpha$-approximation randomized algorithm for \match\ in the message-passing model with the error probability of at most $1/4$ has the communication complexity of $\Omega(\alpha^2kn)$ bits. 
Moreover, this communication complexity holds for an instance of a bipartite graph.
\end{theorem}

In this paper we are more interested in the case when $k\gg \log n$, since otherwise the trivial lower bound of $\Omega(n\log n)$ bits (the number of bits to describe a maximum matching) is already near-optimal.

For \match, a seemingly weaker requirement is that, at the end of the computation, each site $p^{i}$ outputs a set of edges $M^{i}\subseteq E^{i}$ such that $M^{1}\cup M^2 \cup\cdots\cup M^{k}$ is a matching of size that is at least factor $\alpha$ of a maximum matching. However, given such an algorithm, each site might just send $M^{i}$ to the coordinator after running the algorithm, which will increase the total communication cost by at most an additive term of $n$. Therefore, our lower bound also holds for this setting. 

A simple greedy distributed algorithm solves \match\ for $\alpha = 1/2$ with the communication cost of $O(kn \log n)$ bits. This algorithm is based on computing a maximal matching in graph $G$. A maximal matching is a matching whose size cannot be enlarged by adding one or more edges. A maximal matching is computed using a greedy sequential procedure defined as follows. Let $G(E')$ be the graph induced by a subset of edges $E'\subseteq E$. Site $p^1$ computes a maximal matching $M^1$ in $G(E^1)$, and sends it to $p^2$ via the coordinator. Site $p^2$ then computes a maximal matching $M^2$ in $G(E^1\cap E^2)$ by greedily adding edges in $E^2$ to $M^1$, and then sends $M^2$ to site $p^3$. This procedure is continued and it is completed once site $p^k$ computed $M^k$ and sent it to the coordinator. Notice that $M^k$ is a maximal matching in graph $G$, hence it is a $1/2$-approximation of a maximum matching in $G$. The communication cost of this protocol is $O(kn\log n)$ bits because the size of each $M^i$ is at most $n$ edges and each edge's identifier can be encoded with $O(\log n)$ bits. This shows that our lower bound is tight up to a $\log n$ factor. This protocol is essentially sequential and takes $O(k)$ rounds in total. We show that Luby's classic parallel algorithm for maximal matching~\cite{luby1986simple} can be easily adapted to our model with $O(\log n)$ rounds of computation and $O(k n\log^2 n)$ bits of communication. 

In Section~\ref{sec:up}, we show that our lower bound is also tight with respect to the approximation ratio parameter $\alpha$ for any $0 < \alpha \le 1/2$ up to a $\log n$ factor. It was shown  in \cite{WZ13} that many statistical estimation problems and graph combinatorial problems require $\Omega(kn)$ bits of communication to obtain an {\em exact} solution. Our lower bound shows that for \match\ even computing a constant approximation requires this amount of communication.

The lower bound established in this paper applies also more generally for a broader range of graph combinatorial problems. Since a bipartite maximum matching problem can be found by solving a {\em max-flow} problem, our lower bound also holds for \emph{approximate max-flow}. Our lower bound also implies a lower bound for \emph{graph sparsification} problem; see \cite{AGM12} for definition. This is because in our lower bound construction (see Section~\ref{sec:lb}), the bipartite graph under consideration contains many cuts of size $\Theta(1)$ which have to be included in any sparsifier. By our construction, these edges form a good approximate maximum matching, and thus any good sparsifier recovers a good matching. In~\cite{AGM12}, it was shown that there is a sketch-based $O(1)$-approximate graph sparsification algorithm with the sketch size of $\tilde{O}(n)$ bits, which directly translates to an approximation algorithm of $\tilde{O}(kn)$ communication in our model. Thus, our lower bound is tight up to a poly-logarithmic factor for the graph sparsification problem.

We briefly discuss the main ideas and techniques of our proof of the lower bound for \match. As a hard instance, we use a bipartite graph $G = (U, V, E)$ with $\abs{U} = \abs{V} = n/2$. Each site $p^i$ holds a set of $r = n/(2k)$ vertices which is a partition of the set of left vertices $U$. The neighbors of each vertex in $U$ is determined by a two-party set-disjointness instance (\DISJ, defined formally in Section~\ref{sec:DISJ}). There are in total $rk = n/2$ \DISJ\ instances, and we want to perform a direct-sum type of argument on these $n/2$ \DISJ\ instances. We show that due to symmetry, the answer of \DISJ\ can be recovered from a reported matching, and then use information complexity to establish the direct-sum theorem. For this purpose, we use a new definition of the information cost of a protocol in the message-passing model. 

We believe that our techniques would prove useful to establish the communication complexity for other graph combinatorial problems in the message-passing model. The reason is that for many graph problems whose solution certificates ``span" the whole graph (e.g., connected components, vertex cover, dominating set, etc), it is natural that a hard instance would be like for the maximum matching problem, i.e., each of the $k$ sites would hold roughly $n/k$ vertices and the neighbourhood of each vertex would define an independent instance of a two-party communication problem. 

\subsection{Related work}

% Book account on matching theory \cite{lovasz86}. 

The problem of finding an approximate maximum matching in a graph has been studied for various computation models, including the streaming computation model~\cite{AMS99}, MapReduce computation model~\cite{KSV10,GSZ11}, and a traditional distributed computation model known as $\mathsf{LOCAL}$ computation model.  

In \cite{open}, the maximum matching was presented as one of open problems in the streaming computation model. Many results have been established since then by various authors \cite{AG11}, \cite{AG11b}, \cite{AGM12b}, \cite{AKLY16}, \cite{ELMS11}, \cite{KonradMM12}, \cite{Konrad15}, \cite{K13}, \cite{KKS14}, \cite{M05}, and \cite{Z12}. Many of the studies were concerned with a streaming computation model that allows for $\tilde O(n)$ space; referred to as the semi-streaming computation model. The algorithms developed for the semi-streaming computation model can be directly applied to obtain a constant-factor approximation of maximum matching in a graph in the message-passing model that has a communication cost of $\tilde{O}(kn)$ bits.  

For approximate maximum matching problem in the MapReduce model, \cite{LMSV11} gave a $1/2$-approximation algorithm, which requires a constant number of rounds and uses $\tilde{O}(m)$ bits of communication, for any input graph with $m$ edges.  

The approximate maximum matching has been studied in the $\mathsf{LOCAL}$ computation model by various authors \cite{israeli1986fast,LPP08,LPR07,WW04}. In this computation model, each processor corresponds to a unique vertex of the graph and edges represent bidirectional communications between processors. The time advances over synchronous rounds. In each round, every processor sends a message to each of its neighbours, and then each processor performs a local computation using as input its local state and the received messages. Notice that in this model, the input graph and the communication topology are the same, while in the message-passing model the communication topology is essentially a complete graph which is different from the input graph and, in general, sites do not correspond to vertices of the topology graph.  
 
A variety of graph and statistical computation problems have been recently studied in the message-passing model~\cite{KNPR13}, \cite{PVZ12}, \cite{WZ12}, \cite{WZ13}, \cite{WZ14}. A wide range of graph and statistical problems has been shown to be hard in the sense of requiring $\Omega(kn)$ bits of communication, including \emph{graph connectivity}~\cite{PVZ12,WZ13}, \emph{exact counting of distinct elements}~\cite{WZ13}, and \emph{$k$-party set-disjointness}~\cite{BEOPV13}. Some of these problems have been shown to be hard even for random order inputs~\cite{KNPR13}. 

In \cite{BEOPV13}, it has been shown that the communication complexity of the \emph{$k$-party set-disjointness} problem in the message-passing model is $\Omega(kn)$ bits. This work was independent and concurrent to ours. Incidentally, it uses a similar but different input distribution to ours. Similar input distributions were also used in previous work such as \cite{PVZ12} and \cite{WZ12}. This is not surprising because of the nature of the message-passing model. There may exist a reduction between the $k$-party set-disjointness and \match\, but showing this is non-trivial and would require a formal proof. The proof of our lower bound is different in that we use a reduction of the $k$-party \match\ to a $2$-party set-disjointness using a symmetrisation argument, while \cite{BEOPV13} uses a coordinative-wise direct-sum theorem to reduce the $k$-party set-disjointness to a $k$-party $1$-bit problem. 

The approximate maximum matching has been recently studied in the coordinator model under additional condition that the sites send messages to the coordinator simultaneously and once, referred to as the simultaneous-communication model. The coordinator then needs to report the output that is computed using as input the received messages. It has been shown in \cite{AKLY16} that for the vertex partition model, our lower bound is achievable by a simultaneous protocol for any $\alpha \le 1/\sqrt{k}$ up to a poly-logarithmic factor.

The communication/round complexity of approximate maximum matching has been studied in the context of finding efficient economic allocations of items to agents, in markets that consist of unit-demand agents in a distributed information model where agents' valuations are unknown to a central planner, which requires communication to determine an efficient allocation. This amounts to studying the communication or round complexity of approximate maximum matching in a bipartite graph that defines preferences of agents over items. In a market with $n$ agents and $n$ items, this amounts to approximate maximum matching in the $n$-party model with a left vertex partition. \cite{DNO14} and \cite{ANRW15} studied this problem in the so called blackboard communication model, where messages sent by agents can be seen by all agents.
For one-round protocols, \cite{DNO14} established a tight trade-off between message size and approximation ratio. As indicated by the authors in~\cite{DNO14}, their randomized lower bound is actually a special case of ours. In a follow-up work, \cite{ANRW15} obtained the first non-trivial lower bound on the number of rounds for general randomized protocols.

%\cite{GKS16} permutation-invariant functions

%\cite{Duris98} multi-party communication complexity

\subsection{Roadmap}

In Section~\ref{sec:preliminary} we present some basic concepts of probability and information theory, communication and information complexity that are used throughout the paper. Section~\ref{sec:lb} presents the lower bound and its proof, which is the main result of this paper. Section~\ref{sec:up} establishes the tightness of the lower bound up to a poly-logarithmic factor. Finally, in Section~\ref{sec:conc}, we conclude.

\section{Preliminaries}\label{sec:preliminary}

\subsection{Basic facts and notation}
 
Let $[q]$ denote the set $\{1, 2, \ldots, q\}$, for given integer $q \ge 1$. All logarithms are assumed to have base $2$. We use capital letters $X, Y, \ldots$ to denote random variables and the lower case letters $x, y, \ldots$ to denote specific values of respective random variables $X, Y, \ldots$. 

We write $X\sim \mu$ to mean that $X$ is a random variable with distribution $\mu$, and $x \sim \mu$ to mean that $x$ is a sample from distribution $\mu$. For a distribution $\mu$ on a domain ${\mathcal X}\times {\mathcal Y}$, and $(X,Y)\sim \mu$, we write $\mu(x|y)$ to denote the conditional distribution of $X$ given $Y = y$.  

For any given probability distribution $\mu$ and positive integer $t\ge 1$, we denote with $\mu^t$ the $t$-fold product distribution of $\mu$, i.e. the distribution of $t$ independent and identically distributed random variables according to distribution $\mu$.

We will use the following distances between two probability distributions $\mu$ and $\nu$ on a discrete set ${\mathcal X}$: (a) the \emph{total variation distance} defined as
$$
d(\mu,\nu) = \frac{1}{2}\sum_{x\in {\mathcal X}} |\mu(x) - \nu(x)| = \max_{S\subseteq {\mathcal X}}|\mu(S)-\nu(S)|
$$
and, (b) the \emph{Hellinger distance} defined as
$$
h(\mu,\nu) = \sqrt{\frac{1}{2}\sum_{x\in {\mathcal X}} \left(\sqrt{\mu(x)} - \sqrt{\nu(x)}\right)^2}.
$$
The total variation distance and Hellinger distance satisfy the following relation:

\begin{lemma} For any two probability distributions $\mu$ and $\nu$, the total variation distance and the Hellinger distance between $\mu$ and $\nu$ satisfy
$$
d(\mu,\nu) \le \sqrt{2} h(\mu,\nu).
$$
\label{lem:tothell}
\end{lemma}

With a slight abuse of notation for two random variables $X\sim \mu$ and $Y\sim \nu$, we write $d(X,Y)$ and $h(X,Y)$ in lieu of $d(\mu,\nu)$ and $h(\mu,\nu)$, respectively.

We will use the the following two well-known inequalities. 

\paragraph{Hoeffding's inequality} Let $X$ be the sum of $t\ge 1$ independent and identically distributed random variables that take values in $[0,1]$. Then, for any $s\ge 0$,
$$
\Pr[X -\E[X] \ge s]\le e^{-2s^2/t}.
$$

\paragraph{Chebyshev's inequality} Let $X$ be a random variable with variance $\sigma^2 > 0$. Then, for any $s > 0$,
$$
\Pr[|X-\E[X]|\ge s] \le \frac{\sigma^2}{s^2}. 
$$

% players vs site

\subsection{Information theory}

For two random variables $X$ and $Y$, let $H(X)$ denote the \emph{Shannon entropy} of the random variable $X$, and let $H(X|Y)=\E_y[H(X|Y=y)]$ denote the \emph{conditional entropy} of $X$ given $Y$. Let $I(X;Y)=H(X)-H(X|Y)$ denote the \emph{mutual information} between $X$ and $Y$, and let $I(X;Y|Z) = H(X|Z)- H(X|Y,Z)$ denote the \emph{conditional mutual information} given $Z$. The mutual information between any $X$ and $Y$ is non negative, i.e. $I(X; Y) \ge 0$, or equivalently, $H(X|Y)\le H(X)$. 

We will use the following relations from the information theory:

\paragraph{Chain rule for mutual information} For any jointly distributed random variables $X^1, X^2, \ldots, X^t$, $Y$ and $Z$,
$$
I(X^1,X^2,\ldots,X^t; Y|Z) = \sum_{i=1}^t I(X_i;Y|X_1,\ldots,X_{i-1},Z). 
$$

\paragraph{Data processing inequality} If $X$ and $Z$ are conditionally independent random variables given $Y$, then 
$$
I(X; Y| Z)\le I(X; Y) \hbox{ and } I(X; Z)\le I(X; Y).
$$

\paragraph{Super-additivity of mutual information} If $X^1,X^2, \ldots, X^t$ are independent random variables, then 
$$
I(X^1,X^2, \ldots,  X^t; Y)\ge \sum_{i=1}^t I(X^i; Y).
$$

\paragraph{Sub-additivity of mutual information} If $X^1,X^2, \ldots, X^t$ are conditionally independent random variables given $Y$, then
$$
I(X^1,X^2, \ldots, X^t; Y) \le \sum_{i=1}^t I(X^i; Y).
$$

\subsection{Communication complexity}

In the two party communication complexity model two players, Alice and Bob, are required to jointly compute a function $f:{\mathcal X}\times {\mathcal Y} \rightarrow {\mathcal Z}$. Alice is given $x\in \mathcal{X}$ and Bob is given $y\in \mathcal{Y}$, and they want to jointly compute the value of $f(x,y)$ by exchanging messages according to a randomized protocol $\Pi$. 

We use $\Pi_{xy}$ to denote the \emph{random transcript} (i.e., the concatenation of messages) when Alice and Bob run $\Pi$ on the input $(x,y)$, and $\Pi(x,y)$ to denote the \emph{output} of the protocol. When the input $(x,y)$ is clear from the context, we will simply use $\Pi$ to denote the transcript. We say that $\Pi$ is a \emph{$\gamma$-error protocol} if for every input $(x,y)$, the probability that $\Pi(x,y)\neq f(x,y)$ is not larger than $\gamma$, where the probability is over the randomness used in $\Pi$. We will refer to this type of error as \emph{worst-case error}. An alternative and weaker type of error is the \emph{distributional error}, which is defined analogously for an input distribution, and where the error probability is over both the randomness used in the protocol and the input distribution.

Let $\abs{\Pi_{xy}}$ denote the length of the transcript in information bits. The communication cost of $\Pi$ is 
$$
\max_{x,y} \abs{\Pi_{xy}}.
$$ 
The \emph{$\gamma$-error randomized communication complexity} of $f$, denoted by $R_\gamma(f)$, is the minimal cost of any $\gamma$-error protocol for $f$.

The multi-party communication complexity model is a natural generalization to $k\ge 2$ parties, where each party has a part of the input, and the parties are required to jointly compute a function $f:{\mathcal X}^k \rightarrow {\mathcal Z}$ by exchanging messages according to a randomized protocol. 

For more information about communication complexity, we refer the reader to \cite{kushilevitz1997communication}.

\subsection{Information complexity}
\label{sec:ic-intro}

The \emph{communication complexity} quantifies the number of bits that need to be exchanged by two or more players in order to compute some function together, while the \emph{information complexity} quantifies the amount of information of the inputs that must be revealed by the protocol. The information complexity has been extensively studied in the last decade, e.g., \cite{Chakrabarti01,BYJKS02,barak2010compress,WZ12,braverman2012interactive}. There are several definitions of information complexity. In this paper, we follow the definition used in \cite{BYJKS02}. In the two-party case, let $\mu$ be a distribution on $\mathcal{X}\times \mathcal{Y}$, we define the information cost of $\Pi$ measured under $\mu$ as
$$
IC_{\mu}(\Pi)=I(X, Y ; \Pi_{XY} | R)
$$
where $(X,Y)\sim \mu$ and $R$ is the public randomness used in $\Pi$. For notational convenience, we will omit the subscript of $\Pi_{XY}$ and simply use $I(X,Y;\Pi | R)$ to denote the information cost of $\Pi$.  It should be clear that $IC_{\mu}(\Pi)$ is a function of $\mu$ for a fixed protocol $\Pi$. Intuitively, this measures how much information of $X$ and $Y$ is revealed by the transcript $\Pi_{XY}$. For any function $f$, we define the information complexity of $f$ parametrized by $\mu$ and $\gamma$ as
$$
IC_{\mu,\gamma}(f)=\min_{\gamma\hbox{-error } \Pi} IC_\mu(\Pi).
$$

\subsection{Information complexity and coordinator model}
\label{sec:ic-coordinator}

We can indeed extend the above definition of information complexity to $k$-party coordinator model. That is, let $X^i$ be the input of player $i$ with $(X^1, \ldots, X^k) \sim \mu$ and $\Pi$ be the whole transcript, then we could define $IC_{\mu}(\Pi)=I(X^1, X^2,\ldots, X^k;\Pi | R)$. However, such a definition does not fully explore the point-to-point communication feature of the coordinator model. Indeed, the lower bound we can prove using such a definition is at most what we can prove under the blackboard model and our problem admits a simple algorithm with communication $O(n \log n + k)$ in the blackboard model. 
In this paper we give a new definition of information complexity for the coordinator model, which allows us to prove higher lower bounds compared with the simple generalization. Let $\Pi^i$ be the transcript between player $i$ and the coordinator, thus $\Pi = \Pi^1 \circ \Pi^2 \circ \ldots \circ \Pi^k$. We define the information cost for a function $f$ with respect to input distribution $\mu$ and the error parameter $\gamma \in [0,1]$ in the coordinator model as 
$$
IC_{\mu, \gamma}(f) = \min_{\gamma\hbox{-error} \Pi} \sum_{i=1}^k I(X^1, X^2,\cdots,  X^k; \Pi^i).
$$

\begin{theorem}
\label{thm:R-IC}
$R_{\gamma}(f) \ge IC_{\mu, \gamma}(f)$ for any distribution $\mu$.
\end{theorem}
%In general, $\mu$, under which the information cost is measured, and $\nu$ can be different, so we will make explicit about $\mu$ and $\nu$ when talking about information complexity unless they are clear from the context.

\begin{proof}
For any protocol $\Pi$, the expected size of its transcript is (we abuse the notation by using $\Pi$ also for the transcript)
$\E[\abs{\Pi}] = \sum_{i=1}^k \E[\abs{\Pi^i}]\ge \sum_{i=1}^k H(\Pi^i)\ge IC_{\mu,\gamma}(\Pi).$
The theorem then follows because the worst-case communication cost is at least the average-case communication cost.
\end{proof}

\begin{lemma}
\label{lem:IC-in-MP}
If $Y$ is independent of the random coins used by the protocol $\Pi$, then
$$
IC_{\mu, \gamma}(f) \ge \min_{\Pi} \sum_{i=1}^k I(X^i, Y; \Pi^i).
$$
\end{lemma} 

\begin{proof} The statement directly follows from the data processing inequality because given $X^1, X^2, \ldots, X^k$, $\Pi$ is fully determined by the random coins used, and is thus independent of $Y$.
\end{proof}

%In this paper we will measure information cost with respect to distributional errors. Given an input distribution $\nu$, we say a protocol has distributional error $\gamma$ under $\nu$ if the protocol errors with probability at most $\gamma$, where the probability is taken over both the randomness used in the protocol and the input distribution $\nu$. Clearly any lower bounds proved for the distributional error also hold for the worst-case error. 

\section{Lower Bound}\label{sec:lb}
%In this section we prove our main result, the lower bound in Theorem~\ref{thm:main}. We first given an outline that conveys the intuition and main ideas.   

The lower bound is established by constructing a hard distribution for the input bipartite graph $G=(U,V, E)$ such that $\abs{U} = \abs{V} = n/2$. 

We first discuss the special case when the number of sites $k$ is equal to $n/2$, and each site is assigned one unique vertex in $U$ together with all its adjacent edges. We later discuss the general case.

A natural approach to approximately compute a maximum matching in a graph is to randomly sample a few edges from each site, and hope that we can find a good matching using these edges. To rule out such strategies, we construct random edges as follows.

We create a large number of \emph{noisy edges} by randomly picking a small set of nodes $V_0\subseteq V$ of size roughly $\alpha n/10$ and connect each node in $U$ to each node in $V_0$ independently at random with a constant probability. Note that there are $\Theta(\alpha n^2)$ such edges and the size of any matching that can be formed by these edges is at most $\alpha n/10$, which we will show to be asymptotically $\frac{\alpha}{2}\mathrm{OPT}$, where OPT is the size of a maximum matching. 

We next create a set of \emph{important edges} between $U$ and $V_1=V\setminus V_0$ such that each node in $U$ is adjacent to at most one random node in $V_1$. These edges are important in the sense that although there are only $\Theta(|U|) = \Theta(n)$ of them, the size of a maximum matching they can form is large, of the order $\mathrm{OPT}$. Therefore, to compute a matching of size at least $\alpha\mathrm{OPT}$, it is necessary to find and include $\Theta(\alpha\mathrm{OPT}) = \Theta(\alpha n)$ important edges. 

We then show that finding an important edge is in some sense equivalent to solving a set-disjointness (\DISJ) instance, and thus, we have to solve $\Theta(n)$ \DISJ\ instances. The concrete implementation of this intuition is via an embedding argument. 

In the general case, we create $n/(2k)$ independent copies of the above random bipartite graph, each with $2k$ vertices, and assign $n/(2k)$ vertices to each site (one from each copy). We then prove a direct-sum theorem using information complexity.

In the following, we introduce the two-party \AND\ problem and the two-party \DISJ\ problem. These two problems have been widely studied and tight bounds are known (e.g.~\cite{BYJKS02}). For our purpose, we need to prove stronger lower bounds for them. We then give a reduction from \DISJ\ to \match\ and prove an information cost lower bound for \match\ in Section~\ref{sec:matching}.

%This is not entirely accurate -- in fact, we do not embed $k$ instances of \DISJ\ but use the symmetrization technique instead -- but it illustrates the main idea.

\subsection{The two-party \AND\ problem}
\label{sec:AND}

In the two-party \AND\ communication problem, Alice and Bob hold bits $a$ and $b$ respectively, and they want to compute the value of the function \AND$(a, b) = a \wedge b$.
 
Next we define input distributions for this problem. Let $A,B$ be random variables corresponding to the inputs of Alice and Bob respectively. Let $p \in (0,1/2]$ be a parameter. Let $\tau_q$ denote the probability distribution of a Bernoulli random variable that takes value $0$ with probability $q$ or value $1$ with probability $1-q$. We define two input probability distributions $\nu$ and $\mu$ for $(A, B)$ as follows. 

\begin{enumerate}
\item[$\nu$:] Sample $w\sim \tau_p$, and then set the value of $(a,b)$ as follows: if $w = 0$, let $a\sim \tau_{1/2}$ and $b = 0$; otherwise, if $w = 1$, let $a = 0$, and $b \sim \tau_{p}$. Thus, we have
\begin{eqnarray*}
\begin{array}{l}
(A, B)  =  \left\{
  \begin{array}{rl}
   (0, 0) & \quad \hbox{w. p.} \quad  p(3 - 2p)/2 \\
   (0, 1) & \quad \hbox{w. p.} \quad (1-p)^2 \\
   (1, 0) & \quad \hbox{w. p.} \quad p/2
  \end{array}
  \right. .
\end{array}
\end{eqnarray*}

\item[$\mu$:] Sample $w\sim \tau_p$, and then choose $(a,b)$ as above (i.e. sample $(a,b)$ according to $\nu$). Then, reset the value of $a$ to be $0$ or $1$ with equal probability (i.e. set $a\sim \tau_{1/2}$).  
\end{enumerate}

Here $w$ is an axillary random variable to break the dependence of $A$ and $B$, as we can see $A$ and $B$ are not independent, but conditionally independent given $w$.
Let $\delta$ be the probability that $(A, B) = (1, 1)$ under distribution $\mu$, which is $(1-p)^2/2$.

For the special case $p = 1/2$, by~\cite{BYJKS02}, it is shown that, for any \emph{private} coin protocol $\Pi$ with worst-case error probability $1/2-\beta$, the information cost
$$
I(A, B; \Pi | W)\ge \Omega(\beta^2)
$$ 
where the information cost is measured with respect to $\nu$ and $W$ is the random variable corresponding to $w$. Note that the above mutual information is different from the definition of information cost; it is referred to as \emph{conditional information cost} in~\cite{BYJKS02}. It is smaller than the standard information cost by data processing inequality ($\Pi$ and $W$ are conditionally independent given $A,B$). For a fixed protocol $\Pi$, the joint probability distribution $(A,B,\Pi,W)$ is determined by the distribution of $(A,B,W)$ and so is $I(A, B; \Pi | W)$. Therefore, when we say the (conditional) information cost is measured w.r.t. $\nu$, it means that the mutual information, $I(A, B; \Pi | W)$, is measured under the joint distribution $(A,B,\Pi,W)$ determined by $\nu$.

The above lower bound might seem counterintuitive, as the answer to \AND\ is always $0$ under the input distribution $\nu$ and a protocol can just output $0$ which does not reveal any information. However, such a protocol will have worst-case error probability $1$, i.e., it is always wrong when the input is $(1,1)$, contradicting the assumption. When distributional error is considered, the (distributional) error and information cost can be measured w.r.t. different input distributions. In our case, the error will be measured under $\mu$ and the information cost will be measured under $\nu$, and we will prove that any protocol having small error under $\mu$ must incur high information cost under $\nu$.

We next derive an extension that generalizes the result of~\cite{BYJKS02} to any $p\in (0,1/2]$ and distributional errors. We will also use the definition of \emph{one-sided error}.

\begin{definition}\label{def:one-sided}
For a two-party binary function $f(x,y)$, we say that a protocol has a \emph{one-sided error} $\gamma$ for $f$ under a distribution if it is always correct when the correct answer is $0$, and is correct with probability at least $1-\gamma$ conditional on $f(x, y) = 1$.
\end{definition}

Recall that $\delta$ is the probability that $(A,B)=(1,1)$ when $(A,B)\sim\mu$, which is $(1-p)^2/2$. Recall that $p\in(0,1/2]$, and thus $\delta\le 1/2$. Note that a distributional error of $\delta$ under $\mu$ is trivial, as a protocol that always outputs $0$ achieves this (but it has one-sided error $1$). Therefore, for two-sided error, we will consider protocols with error probability slightly better than the trivial protocol, i.e., with error probability $\delta-\beta$ for some $\beta\le\delta$. 

\begin{theorem} \label{thm:AND}
Suppose that $\Pi$ is a public coin protocol for \AND\, which has distributional error $\delta - \beta$, for $\beta \in (0, \delta)$, under input distribution $\mu$; let $R$ denote its public randomness. Then 
$$
I(A, B; \Pi | W ,R) = \Omega(p(\beta/\delta)^2) 
$$
where the information is measured with respect to $\nu$.  

If $\Pi$ has a one-sided error $1-\beta$, then 
$$
I(A, B; \Pi|W,R) = \Omega(p\beta).
$$
\end{theorem}
If we set $p=1/2$, the first part of Theorem~\ref{thm:AND} recovers the result of~\cite{BYJKS02}.

\begin{proof}[of Theorem~\ref{thm:AND}]
We will use $\Pi_{ab}$ to denote the transcript when the input is $a,b$. By definition,
\begin{eqnarray}
I(A,B; \Pi_{AB} | W)
&=& p\cdot I(A, 0; \Pi_{A0} | W=0)+  (1-p)\cdot I(0, B; \Pi_{0B} | W=1)\nonumber \\
&=&  p\cdot I(A; \Pi_{A0})+(1-p)\cdot I(B; \Pi_{0B})\label{eqIUV}.
\end{eqnarray}

With a slight abuse of notation, in (\ref{eqIUV}), $A$ and $B$ are random variables with distributions $\tau_{1/2}$ and $\tau_p$, respectively.

For any random variable $U$ with distribution $\tau_{1/2}$, the following two inequalities were established in \cite{BYJKS02}:
\begin{equation}
I(U; \Pi_{U0})\ge h^2(\Pi_{00},\Pi_{10})
\label{equ:Ih1}
\end{equation}
and
\begin{equation}
I(U; \Pi_{0U}) \ge h^2(\Pi_{00}, \Pi_{01})
\label{equ:Ih2}
\end{equation}
 where $h(X,Y)$ is the Hellinger distance between two random variables $X$ and $Y$. 

We can apply these bounds to lower bound the term $I(A; \Pi_{A0})$. However, we cannot apply them to lower bound the term $I(B; \Pi_{0B})$ when $p < 1/2$ because then the distribution of $B$ is not $\tau_{1/2}$. To lower bound the term $I(B; \Pi_{0B})$, we will use the following well-known property, whose proof can be found in the book \cite{cover2006} (Theorem 2.7.4).
 
\begin{lemma}
Let $(X,Y)\sim p(x,y)=p(x)p(y|x)$. The mutual information $I(X,Y)$ is a concave function of $p(x)$ for fixed $p(y|x)$.
\label{lem:mutinf}
\end{lemma}

%In our case, $x$ is $B$ and $y$ is $\Pi_{0B}$, and it is easy to see that the conditional probability $\Pr[\Pi_{0B}=\pi|B=b]$ is fixed for any $\pi$ and $b$. 
Hence, the mutual information $I(B;\Pi_{0B})$ is a concave function of the distribution $\tau_p$ of $B$, since the distribution of $\Pi_{0B}$ is fixed given $B$. 

Recall that $\tau_p$ is the probability distribution that takes value $0$ with probability $p$ and takes value $1$ with probability $1-p$.
Note that $\tau_p$ can be expressed as a convex combination of $\tau_{1/2}$ and $\tau_0$ (always taking value $1$) as follows: $\tau_p = 2p \tau_{1/2} + (1-2p) \tau_0$. (Recall that $p$ is assumed to be smaller than $1/2$.) Let $B_0\sim \tau_{1/2}$ and $B_1 \sim \tau_0$. Then, using Lemma~\ref{lem:mutinf}, we have 
\begin{eqnarray*}
I(B; \Pi_{0B}) &\ge & 2p\cdot I(B_0; \Pi_{0B_0})+(1-2p)\cdot I(B_1; \Pi_{0B_1})\\
&& \ge 2p\cdot h^2(\Pi_{00},\Pi_{01})
\end{eqnarray*}  
where the last inequality holds by (\ref{equ:Ih2}) and non-negativity of mutual information. 

Thus, we have
\begin{eqnarray}
 I(A,B; \Pi_{AB} | W)
&=& p\cdot I(A; \Pi_{A0})+(1-p)\cdot I(B; \Pi_{0B}) \nonumber\\
&\ge& p \cdot h^2(\Pi_{00},\Pi_{10})+(1-p)2p\cdot h^2(\Pi_{00},\Pi_{01})\nonumber\\
&\ge& p \cdot( h^2(\Pi_{00},\Pi_{10})+ h^2(\Pi_{00},\Pi_{01})) 
\label{laststep}
\end{eqnarray}
where the last inequality holds because $p\le 1/2$.

We next show that if $\Pi$ is a protocol with error probability smaller than or equal to $\delta-\beta$ under distribution $\mu$, then
$$
h^2(\Pi_{00},\Pi_{10})+ h^2(\Pi_{00},\Pi_{01}) = \Omega((\beta/\delta)^2),
$$
which together with other above relations implies the first part of the theorem.

By the triangle inequality,
\begin{eqnarray}
h(\Pi_{00},\Pi_{10})+ h(\Pi_{00},\Pi_{01}) & \ge & h(\Pi_{01},\Pi_{10})\nonumber\\
&=&h(\Pi_{00},\Pi_{11})\label{eq:Hellinger1}
\end{eqnarray}
where the last equality is from the {\em cut-and-paste lemma} in \cite{BYJKS02} (Lemma 6.3).

Thus, we have
\begin{eqnarray}
 h(\Pi_{00},\Pi_{10})+ h(\Pi_{00},\Pi_{01})
&\ge& \frac{1}{2}h(\Pi_{00},\Pi_{10})+\frac{1}{2}(h(\Pi_{00},\Pi_{10})+ h(\Pi_{00},\Pi_{01}))\nonumber\\
&\ge& \frac{1}{2}(h(\Pi_{00},\Pi_{10})+h(\Pi_{00},\Pi_{11}))\nonumber\\
&\ge& \frac{1}{2} h(\Pi_{10},\Pi_{11})\label{eq:Hellinger2}
\end{eqnarray}
where the last inequality is by the triangle inequality.

Similarly,  it holds that
\begin{equation}
h(\Pi_{00},\Pi_{10})+ h(\Pi_{00},\Pi_{01})\ge \frac{1}{2} h(\Pi_{01},\Pi_{11}).\label{eq:Hellinger3}
\end{equation}

From~(\ref{eq:Hellinger1}), (\ref{eq:Hellinger2}) and (\ref{eq:Hellinger3}), for any positive real numbers $a$, $b$, and $c$ such that $a+b+c = 1$, we have
\begin{eqnarray}
 h(\Pi_{00},\Pi_{10})+ h(\Pi_{00},\Pi_{01})
& \ge & \frac{1}{2} (a \cdot h(\Pi_{00},\Pi_{11}) + b\cdot h(\Pi_{01},\Pi_{11}) \nonumber\\
&&+ c\cdot h(\Pi_{10},\Pi_{11})).\label{eqHellcon}
\end{eqnarray}

Let $p^e$ denote the error probability of $\Pi$ and $p^e_{xy}$ denote the error probability of $\Pi$ conditioned on that the input is $(x,y)$. Recall $\delta=\mu(1,1)\le 1/2$. We have 
\begin{eqnarray}
p^e & = & \mu(0,0) p^e_{00}+\mu(1,0)p^e_{10}+\mu(0,1)p^e_{01}+\delta p^e_{11} \nonumber\\
& \ge &\delta \left( \frac{\mu(0,0) p^e_{00}+\mu(1,0)p^e_{10}+\mu(0,1)p^e_{01}}{1-\delta}+ p^e_{11} \right)\nonumber\\
& = & \delta(a^*(p^e_{00}+p^e_{11})+b^*(p^e_{01}+p^e_{11})+c^*(p^e_{10}+p^e_{11})) \label{eq:convex}
\end{eqnarray}
where 
$$
a^* = \frac{\mu(0,0)}{1-\delta}, b^* = \frac{\mu(0,1)}{1-\delta} \hbox{ and } c^* = \frac{\mu(1,0)}{1-\delta},
$$
and clearly $a^*+b^*+c^*=1$.
Let $\Pi(x,y)$ be the output of $\Pi$ when the input is $(x,y)$, which is also a random variable.
Note that 
\begin{eqnarray}
p^e_{00}+p^e_{11}&=& \Pr[\Pi(0,0)=1]+\Pr[\Pi(1,1)=0]\nonumber\\
&=&1-(\Pr[\Pi(0,0)=0]-\Pr[\Pi(1,1)=0])\nonumber\\
&\ge& 1-d(\Pi_{00},\Pi_{11})
\end{eqnarray}
where $d(X,Y)$ denote the total variation distance between probability distributions of random variables $X$ and $Y$. Using Lemma~\ref{lem:tothell}, we have
\begin{equation}
p^e_{00}+p^e_{11} \ge 1 - \sqrt{2}h(\Pi_{00},\Pi_{11}).
\label{equ:pe1}
\end{equation}
By the same arguments, we also have 
\begin{equation}
p^e_{01}+p^e_{11} \ge 1-\sqrt{2}h(\Pi_{01},\Pi_{11})
\label{equ:pe2}
\end{equation}
and
\begin{equation}
p^e_{10}+p^e_{11} \ge 1-\sqrt{2}h(\Pi_{10},\Pi_{11}).
\label{equ:pe3}
\end{equation}

Combining (\ref{equ:pe1}), (\ref{equ:pe2}) and (\ref{equ:pe3}) with (\ref{eq:convex}) and the assumption that $p^e \le \delta - \beta$, we obtain
\begin{eqnarray*}
a^* h(\Pi_{00},\Pi_{11}) + b^* h(\Pi_{10},\Pi_{11})+ c^* h(\Pi_{01},\Pi_{11}) &\ge& \frac{\beta}{\sqrt{2}\delta}.
\end{eqnarray*}

By (\ref{eqHellcon}), we have
$$
h(\Pi_{00},\Pi_{10})+ h(\Pi_{00},\Pi_{01})\ge \frac{\beta}{2\sqrt{2}\delta}.
$$
From the Cauchy-Schwartz inequality, it follows 
\begin{eqnarray*}
 h^2(\Pi_{00},\Pi_{10})+ h^2(\Pi_{00},\Pi_{01})
& \ge & \frac{1}{2} (h(\Pi_{00},\Pi_{10})+ h(\Pi_{00},\Pi_{01}))^2.
\end{eqnarray*}

Hence, we have
$$
h^2(\Pi_{00},\Pi_{10})+ h^2(\Pi_{00},\Pi_{01}) \ge \frac{\beta^2}{16\delta^2}
$$
which combined with (\ref{laststep}) establishes the first part of the theorem.

We now go on to prove the second part of the theorem. Assume $\Pi$ has a one-sided error $1-\beta$, i.e., it outputs $1$ with probability at least $\beta$ if the input is $(1,1)$, and always output correctly otherwise. To boost the success probability, we can run $m$ parallel instances of the protocol and answer $1$ if and only if there exists one instance which outputs $1$. Let $\Pi'$ be this new protocol, and it is easy to see that it has a one-sided error of $(1-\beta)^m$. By setting $m=O(1/\beta)$, it is at most $1/10$, and thus the (two-sided) distributional error of $\Pi'$ under $\mu$ is smaller than $\delta/10$. By the first part of the theorem, we know $I(A, B; \Pi' | W) = \Omega(p)$. We also have
\begin{eqnarray}
I(A, B; \Pi' | W)&=& I(A, B; \Pi_1,\Pi_2,\ldots, \Pi_m | W) \nonumber \\
&\le& \sum_{i=1}^m I(A, B; \Pi_i | W) \nonumber\\
&=& m I(A, B; \Pi | W), \nonumber
\end{eqnarray}
where the inequality follows from the sub-additivity and the fact that $\Pi_1, \Pi_2,\ldots, \Pi_m$ are conditionally independent of each other given $A, B$ and $W$. Thus, we have $I(A, B; \Pi | W)\ge \Omega(p/m) = \Omega(p\beta)$.
\end{proof}

\subsection{The two-party \DISJ\ communication problem}
\label{sec:DISJ}
The two-party \DISJ\  communication problem with two players, Alice and Bob, who hold strings of bits $x = (x_1,x_2, \ldots, x_k)$ and $y = (y_1, y,\ldots, y_k)$, respectively, and they want to compute 
$$
\hbox{\DISJ}(x,y) = \hbox{\AND}(x_1,y_1)\vee \cdots \vee \hbox{\AND}(x_k,y_k). 
$$
By interpreting $x$ and $y$ as indicator vectors that specify subsets of $[k]$, \DISJ$(x,y)=1$ if and only if the two sets represented by $x$ and $y$ are disjoint. Note that this accommodates the \AND\  problem as a special case when $k = 1$. 

Let $A = (a_1,a_2, \ldots, a_k)$ be Alice's input and $B = (b_1, b_2,\ldots,b_k)$ be Bob's input. We define two input distributions $\nu_k$ and $\mu_k$ for $(A,B)$ as follows. 

\begin{enumerate}
\item[$\nu_k$:] For each $i\in[k]$, independently sample $(a_i,b_i)\sim \nu$, and let $w_i$ be the corresponding auxiliary random variable (see the definition of $\nu$). Define $w=(w_1,w_2,\cdots, w_k)$. 

\item[$\mu_k$:] Let $(a,b)\sim \nu_k$, then pick $d$ uniformly at random from $[k]$, and reset $a_d$ to be $0$ or $1$ with equal probability. Note that $(a_d, b_d) \sim \mu$, and the probability that \DISJ$(A,B)=1$ is equal to $\delta$. 
\end{enumerate}

We define the one-sided error for \DISJ\ similarly: A protocol has a one-sided error $\delta$ for \DISJ\ if it is always correct when \DISJ$(x,y) = 0$, and is correct with probability at least $1-\delta$ when \DISJ$(x,y) = 1$.
 
\begin{theorem}
\label{thm:DISJ}
Let $\Pi$ be any public coin protocol for \DISJ\ with error probability $\delta - \beta$ on input distribution $\mu_k$, where $\beta \in (0, \delta)$, and let $R$ denote the public randomness used by the protocol. Then 
$$
I(A,B; \Pi | W, R) = \Omega(kp (\beta/\delta)^2)
$$
where the information is measured w.r.t. $\mu_k$. 

If $\Pi$ has a one-sided error $1-\beta$, then 
$$
I(A, B; \Pi | W, R) = \Omega(k p \beta).
$$
\end{theorem}

\begin{proof} We first consider the two-sided error case. Let $\Pi$ be a protocol for \DISJ\ with distributional error $\delta-\beta$ under $\mu_k$. Consider the following reduction from \AND\ to \DISJ.  Alice has input $u$, and Bob has input $v$. They want to decide the value of $u\wedge v$. They first publicly sample $j \in [k]$, and embed $u,v$ in the $j$-th position, i.e. set $a_j=u$ and $b_j=v$. Then they publicly sample $w_{j'}$ according to $\tau_p$ for all $j' \neq j$. Let $w_{-j} = (w_1, \ldots, w_{j-1}, w_{j+1}, \ldots, w_k)$. Conditional on $w_{j'}$, they sample $(a_{j'}, b_{j'})$ such that $(a_j,b_j)\sim \nu$ for each $j' \neq j$. Note that this step can be done using only private randomness, since, in the definition of $\nu$, $a_{j'}$ and $b_{j'}$ are independent given $w_{j'}$. Then they run the protocol $\Pi$ on the input $(a,b)$ and output whatever $\Pi$ outputs. Let $\Pi'$ denote this protocol for \AND. Let $U,V,A,B,W,J$ be the corresponding random variables of $u,v,a,b,w,j$ respectively. It is easy to see that if $(U,V)\sim \mu$, then $(A,B)\sim \mu_k$, and thus the distributional error of $\Pi'$ is $\delta-\beta$ under $\mu$. The public coins used in $\Pi'$ include $J$, $W_{-J}$ and the public coins $R$ of $\Pi$. 

We first analyze the information cost of $\Pi'$ under $(A,B)\sim\nu_k$. We have 
\begin{eqnarray}
 \frac{1}{k} I(A,B;\Pi | W,R)
&\ge& \frac{1}{k} \sum_{j=1}^k I(A_j,B_j;\Pi | W_j, W_{-j},R)\label{equ:first} \\
&=&\frac{1}{k}\sum_{j=1}^k I_{\nu}(U,V;\Pi' | W_j,J=j,W_{-j},R) \label{eq:AndtoDiSJ} \\
&=& I(U, V;\Pi' | W_J,J,W_{-J},R)\nonumber \\
 &=& \Omega(p (\beta/\delta)^2) \label{eq:lasteq}
\label{eq:useand}
\end{eqnarray}
where (\ref{equ:first}) is by the supper-additivity of mutual information, (\ref{eq:AndtoDiSJ}) holds because when $(U,V)\sim\nu$ the conditional distribution of $(U,V,\Pi, W_j,W_{-j},R)$ given $J = j$ is the same as the distribution of $(A_j,B_j,\Pi, W_j, W_{-j},R)$, and (\ref{eq:lasteq}) follows from Theorem~\ref{thm:AND} using the fact that $\Pi'$ has error $\delta-\beta$ under $\mu$. 

We have established that when $(A,B)\sim \nu_k$, it holds
\begin{equation}
I(A,B;\Pi | W,R)=\Omega(kp (\beta/\delta)^2).
\label{equ:IABP}
\end{equation} 

We now consider the information cost when $(A,B)\sim \mu_k$. Recall that to sample from $\mu_k$, we first sample $(a,b)\sim \nu_k$, and then pick $d$ uniformly at random from $[k]$ and reset $a_d$ to $0$ or $1$ with equal probability. Let $\xi$ be the indicator random variable of the event that the last step does not change the value of $a_d$. 

We note that for any jointly distributed random variables $X$, $Y$, $Z$ and $W$,
\begin{equation}
I(X;Y|Z) \ge I(X;Y|Z,W) - H(W).
\label{equ:mut}
\end{equation}
To see this note that by the chain rule for mutual information, we have
$$
I(X,W;Y|Z) = I(X;Y|Z) + I(W;Y|X,Z)
$$
and
$$
I(X,W;Y|Z) = I(W;Y|Z) + I(X;Y|W,Z).
$$
Combining the above two equalities, (\ref{equ:mut}) follows by the facts $I(W;Y|X,Z)\ge 0$ and $I(W;Y|Z) \le H(W|Z) \le H(W)$.

Let $(A,B)\sim \mu_k$ and $(A',B')\sim \nu_k$. We have
\begin{eqnarray*}
I(A,B;\Pi | W,R) &\ge& I(A,B;\Pi | W,R, \xi) - H(\xi)\\
&=& \frac{1}{2} I(A,B;\Pi | W,R, \xi=1) + \frac{1}{2}I(A,B;\Pi | W,R, \xi=0)-1\\
&\ge&  \frac{1}{2}I(A',B';\Pi | W,R)-1\\
&=& \Omega(kp (\beta/\delta)^2)
\end{eqnarray*}
where the first inequality is from (\ref{equ:mut}) and the last equality is by (\ref{equ:IABP}).

The proof for the one-sided error case is the same, except that we use the one-sided error lower bound $\Omega(p\beta)$ in Theorem~\ref{thm:AND} to bound (\ref{eq:useand}).
\end{proof}

\subsection{Proof of Theorem~\ref{thm:main}}
\label{sec:matching}

Here we provide a proof of Theorem~\ref{thm:main}. The proof is based on a reduction of \DISJ\ to \match. We first define the hard input distribution that we use for \match.

The input graph $G$ is assumed to be a random bipartite graph that consists of $r=n/(2k)$ disjoint, independent and identically distributed random bipartite graphs $G^1$, $G^2$, $\ldots$, $G^r$. Each bipartite graph $G^{j} = (U^j, V^j, E^j)$ has the set $U^j = \{u^{j, i}: i\in [k]\}$ of left vertices and the set $V^j = \{v^{j, i}: i\in [k]\}$ of right vertices, both of cardinality $k$. The sets of edges $E^1$, $E^2$, $\ldots$, $E^r$ are defined by a random variable $X$ that takes values in $\{0,1\}^{r\times k \times k}$ such that whether or not $(u^{j,i},v^{j,l})$ is an edge in $E^j$ is indicated by $X_l^{j,i}$.

The distribution of $X$ is defined as follows. Let $Y^1$, $Y^2$, $\ldots$, $Y^r$ be independent and identically distributed random variables with distribution $\mu_k(b)$.\footnote{$\mu_k(b)$ is the marginal distribution of $b$ of the joint distribution $\mu_k$.}. Then, for each $j\in [r]$, conditioned on $Y^j = y^j$, let $X^{j,1}$, $X^{j,2}$, \ldots, $X^{j,k}$ be independent and identically distributed random variables with distribution $\mu_k(a|y^j)$, where $\mu_k(a|y^j)$ is the conditional distribution of $a$ given $b=y^j$. Note that for every $j\in [r]$ and $i\in [k]$, $(X^{j,i},Y^j)\sim \mu_k$.

\begin{figure*}[t]
\begin{center}
\includegraphics[width=0.9\textwidth]{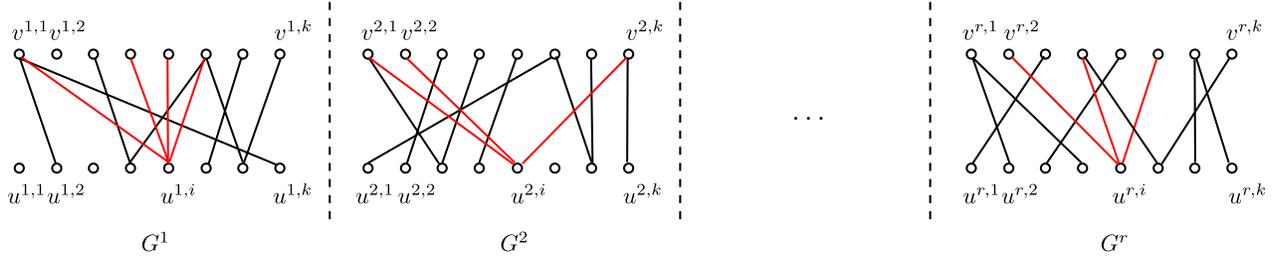}
\end{center}
\caption{Construction of input graph $G$ and its partitioning over sites: $G$ is a composition of bipartite graphs $G^1, G^2, \ldots,G^r$, each having $k$ vertices on each side of the bipartition; each site $i\in [k]$ is assigned edges incident to vertices $u^{1,i}, u^{2,i}, \ldots, u^{r,i}$; the neighbourhood set of vertex $u^{j,i}$ is determined by $X^{j,i}\in \{0,1\}^k$.}
\label{fig:graphpart}
\end{figure*}

We will use the following notation:
$$
X^i = (X^{1,i}, X^{2,i},\ldots,X^{r,i}) \hbox{ for } i\in [k],
$$
and
$$
X = (X^1, X^2, \ldots,X^k),
$$
where each $X^{j,i}\in\{0,1\}^{k}$, and $X^{j,i}_l$ is the $l$th bit.
In addition, we will also use the following notation:
$$
X^{-i} = (X^1, \ldots, X^{i-1},X^{i+1}, \ldots, X^k), \hbox{ for } i\in [k]
$$
and
$$
Y = (Y^1, Y^2, \ldots, Y^r).
$$

Note that $X$ is the input to \match, and $Y$ is {\em not} part of the input for \match, but it is used to construct $X$.

The edge partition of input graph $G$ over $k$ sites $p^1$, $p^2$, $\ldots$, $p^k$ is defined by  assigning all edges incident to vertices $u^{1,i}$, $u^{2,i}$,$\ldots$, $u^{r,i}$ to site $p^i$, or equivalently $p^{i}$ gets $X^{i}$. See Figure~\ref{fig:graphpart} for an illustration.

\paragraph{Input Reduction} Let $a \in \{0,1\}^k$ be Alice's input and $b \in \{0,1\}^k$ be Bob's input for \DISJ. We will first construct an input of \match\ from $(a,b)$, which has the above hard distribution.
In this reduction, in each bipartite graph $G^j$, we carefully embed $k$ instances of \DISJ. The output of a \DISJ\ instance determines whether or not a specific edge in the graph exists. This amounts to a total of $k r = n/2$ \DISJ\ instances embedded in graph $G$. The original input of Alice and Bob is embedded at a random position, and the other $n/2 - 1$ instances are sampled by Alice and Bob using public and private random coins. We then argue that if the original \DISJ\ instance is solved, then with a sufficiently large probability, at least $\Omega(n)$ of the embedded \DISJ\ instances are solved. Intuitively, if a protocol solves an \DISJ\ instance at a random position with high probability, then it should solve many instances at other positions as well, since the input distribution is completely symmetric. We will see that the original \DISJ\ instance can be solved by using any protocol solving \match, the correctness of which also relies on the symmetric property.

Alice and Bob construct an input $X$ for \match\ as follows:

\begin{enumerate}
\item Alice and Bob use public coins to sample an index $I$ from a uniform distribution on $[k]$. Alice constructs the input $X^{I}$ for site $p^I$, and Bob constructs input $X^{-I}$ for other sites (see Figure|~\ref{fig:cor-alice-bob}).

\item Alice and Bob use public coins to sample an index $J$ from a uniform distribution on $[r]$.

\item $G^J$ is sampled as follows: Alice sets $X^{J, I} = a$, and Bob sets $Y^{J} = b$. Bob privately samples 
$$
(X^{J, 1}, \ldots, X^{J,I-1}, X^{J,I}, X^{J,k})\sim \mu_k(a|Y^J)^{k-1}.
$$ 

\item For each $j\in [r]\setminus \{J\}$, $G^j$ is sampled as follows: 
	\begin{enumerate}
	\item Alice and Bob use public coins to sample $W^j = (W_1^j, W_2^j,\ldots, W_k^j) \sim \tau_p^k$. 

	\item Alice and Bob privately sample $X^{j,I}$ and $Y^j$ from $\nu_k(a|W^j)$ and $\nu_k(b|W^j)$, respectively. Bob privately and independently samples 
$$
(X^{j, 1}, \ldots, X^{j,I-1}, X^{j,I}, X^{j,k})\sim \mu_k(a|Y^j)^{k-1}.
$$ 

	\item Alice privately draws an independent sample $d$ from a uniform distribution on $[k]$, and resets $X^{j, I}_{d}$ to $0$ or $1$ with equal probability. As a result, $(X^{j,I},Y^j)\sim \mu_k$. For each $i\in [k]\setminus \{I\}$, Bob privately draws a sample $d$ from a uniform distribution on $[k]$ and resets $X^{j, i}_{d}$ to a sample from $\tau_{1/2}$.
\end{enumerate}
\end{enumerate}

\begin{figure}[h]
\begin{center}
\includegraphics[width=0.35\textwidth]{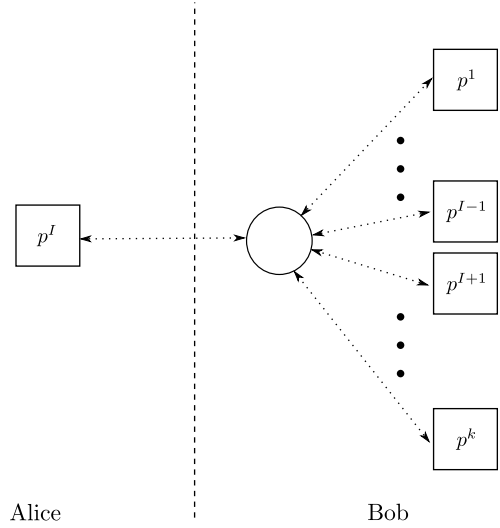}
\end{center}
\caption{Alice simulates site $p^I$ and Bob simulates the rest of the system.}
\label{fig:cor-alice-bob}
\end{figure}

Note that the input $X^{I}$ of site $p^I$ is determined by the public coins, Alice's input $a$ and her private coins. The inputs $X^{-I}$ are determined by the public coins, Bob's input $b$ and his private coins. 

Let $\phi$ denote the distribution of $X$ when $(a,b)$ is chosen according to the distribution $\mu_k$. 

Let $\alpha$ be the approximation ratio parameter. We set $p = \alpha/20 \le 1/20$ in the definition of $\mu_k$.

Given a protocol $\P'$ for \match\ that achieves an $\alpha$-approximation with the error probability at most $1/4$ under $\phi$, we construct a protocol $\P$ for \DISJ\ that has a one-sided error probability of at most $1 - \alpha/10$ as follows.

\paragraph{Protocol $\P$}

\begin{enumerate}
\item Given input $(A, B) \sim \mu_k$, Alice and Bob construct an input $X \sim \phi$ for \match\ as described by the input reduction above. Let $Y = (Y^1,Y^2, \ldots, Y^r)$ be the samples used for the construction of $X$. Let $I, J$ be the two indices sampled by Alice and Bob in the reduction procedure.

\item With Alice simulating site $p~^I$ and Bob simulating other sites and the coordinator, they run $\P'$ on the input defined by $X$. Any communication between site $p^I$ and the coordinator will be exchanged between Alice and Bob. For any communication among other sites and the coordinator, Bob just simulates it without any actual communication. At the end, the coordinator, that is Bob, obtains a matching $M$. 

\item Bob outputs $1$ if, and only if, for some $l \in [k]$, $(u^{J, I}, v^{J,l})$ is an edge in $M$ such that $Y_l^J\equiv B_l = 1$, and $0$, otherwise. 
\end{enumerate}

\paragraph{Correctness} Suppose that \DISJ$(A,B)=0$, i.e., $A_l = 0$ or $B_l = 0$ for all $l \in [k]$. Then, for each $l \in [k]$, we must either have $Y_l^J\equiv B_l = 0$ or $X^{J,I}_l \equiv A_l=0$, but $X^{J,I}_l=0$ means that $(u^{J,I},v^{J,l})$ is not an edge in $M$. Thus, $\P$ will always answer correctly when \DISJ$(A,B)=0$, i.e., it has a one-sided error.

Now suppose that $A_l = B_l = 1$ for some $l \in [k]$. Note that there is at most one such $l$ according to our construction, which we denote by $L$. The output of $\P$ is correct if $(u^{J,I},v^{J,L})$ is an edge in $M$. We next bound the probability of this event.

For each $G^{j}$, for $z\in \{0,1\}$, we let
$$
U_z^j = \{u^{j,i}\in U^j: \hbox{\DISJ}(X^{j, i}, Y^j) = z\},
$$
$$
V_z^j = \{v^{j,i}\in V^j: Y_i^j = z\}
$$
and
$$
U_z = \cup_{j\in [r]} U_z^j \hbox{ and } V_z = \cup_{j\in [r]} V_z^j.
$$

Intuitively, the edges between vertices $U_0 \cup U_1$ and $V_0$ can be regarded as {\em noisy edges} because the total number of such edges is large, but the maximum matching they can form is small (Lemma~\ref{lem:match-property} below). On the other hand, the edges between vertices $U_1$ and $V_1$ can be regarded as {\em important edges} because a maximum matching they can form is large though the total number of such edges is small. Note that there is no edge between vertices $U_0$ and $V_1$. See Figure~\ref{fig:edgetypes} for an illustration. 

%{\bf ZH: Each vertex in $U$ has at most one important edge. To be more accurate, the right most vertex in $U_1$ should not have two neighbors in $V_1$}

\begin{figure}[h]
\begin{center}
\includegraphics[width=0.35\textwidth]{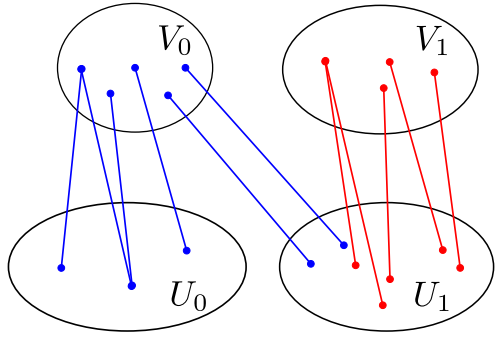}
\end{center}
\caption{Edges between $U_0\cup U_1$ and $V_0$ are noisy edges. Edges between $U_1$ and $V_1$ are important edges. There are no edges between $U_0$ and $V_1$. 
%Each edge corresponds to a \DISJ\ instance where a solid edge indicates an instance with output $1$ and a dashed edge indicates an instance with output $0$. Solid thick edges are the important edges. A good approximate matching has to output many important edges, thus a solid thick edge needs to be in the output matching with sufficiently large probability. The thick edge $(u^{J,I}, v^{J, \ell})$ corresponds to \DISJ$(X^{J,I},Y^J) = $ \DISJ$(s,t)$, that is to the original $2$-party disjointness problem embedded by Alice and Bob. If \DISJ$(s,t) = 1$, then $(u^{J,I}, v^{J, \ell})$ is a solid edge and needs to be included in the output matching with a sufficiently large probability.
}
\label{fig:edgetypes}
\end{figure}

To find a good matching we must choose many edges from the set of important edges. A key property is that all important edges are statistically identical, that is, each important edge is equally likely to be the edge $(u^{J, I}, v^{J, L})$. Thus, $(u^{J, I}, v^{J, L})$ will be included in the matching returned by $\P'$ with a large enough probability. Using this, we can answer whether $X^{J, I}$ and $Y^{J}$ intersect or not, thus, solving the original \DISJ\ problem.

Recall that we set $p=\alpha/20\le 1/20$ and $\delta=(1-p)^2/2$. Thus, $9/20<\delta<1/2$. In the following, we assume $\alpha\ge c/\sqrt{n}$ for some constant, since otherwise the $\Omega(\alpha^2 k n)$ lower bound will be dominated by the trivial lower bound of $k$.\footnote{Since none of the sites can see messages sent by other sites to the coordinator (unless this is communicated by the coordinator), each site needs to communicate with the coordinator at least once to determine the status of the protocol.}

\begin{lemma}
\label{lem:match-property}
With probability at least $1-1/100$, 
$$
\abs{V_0} \le 2pn.
$$ 
\end{lemma}

\begin{proof} Note that each vertex in $\cup_{j\in [r]} V^j$ is included in $V_0$ independently with probability $p(2-p)$. Hence, $\E[\abs{V_0}] = p(2 - p)n/2$, and by the Hoeffding's inequality, we have
\begin{eqnarray*}
\Pr[\abs{V_0} \ge 2pn] & \le & \Pr[\abs{V_0} - \E[\abs{V_0}] \ge pn]\\
& \le & e^{-2 p^2 n}\\
& \le & 1/100.
\end{eqnarray*}
\end{proof}

Notice that Lemma~\ref{lem:match-property} implies that with probability at least $1-1/100$, the size of a maximum matching formed by edges between vertices $V_0$ and $U_0 \cup U_1$ is smaller than or equal to $2pn$.

\begin{lemma}
\label{lem:match-property2}
With probability at least $1-1/100$, the size of a maximum matching in $G$ is at least $n/5$.
\end{lemma}

\begin{proof} Consider the size of a matching in $G^j$ for an arbitrary $j\in [r]$. For each $i \in [k]$, let $L^i$ be the index $l\in [k]$ such that $X_l^{j, i} = Y_l^j = 1$ if such an $l$ exists (note that by our construction at most one such index exists), and let $L^i$ be defined as NULL, otherwise. 

We use a greedy algorithm to construct a matching between vertices $U^j$ and $V^j$. For $i\in [k]$, we connect $u^{j, i}$ and $v^{j, L^i}$ if $L^i$ is not NULL and $v^{j, L^i}$ is not connected to some $u^{j, i'}$ for $i' < i$. The size of such constructed matching is equal to the number of distinct elements in $\{L^1, L^2, \ldots, L^k\}$, which we denote by $R$. We next establish the following claim: 
\begin{equation} \label{cla:binball} 
\Pr[R \ge k/4] \ge 1 - O(1/k). 
\end{equation}

By our construction, we have
$$
\E[\abs{U_1^j}]= \delta k \hbox{ and } \E[\abs{V_1^j}] = (1-p)^2k.
$$
By the Hoeffding's inequality, with probability $1 - e^{-\Omega(k)}$, 
$$
\abs{V_1^j} \ge \frac{9}{10}\E[\abs{V_1^j}] %= \frac{9}{10}(1-p)^2 k 
\ge \frac{4}{5}k
$$
and 
$$
\abs{U_1^j} \ge \frac{9}{10}\E[\abs{U_1^j}] \ge \frac{2}{5}k. 
$$

It follows that with probability $1 - e^{-\Omega(k)}$, it holds that $R$ is at least of value $R'$, where $R'$ is as defined as follows. 

Consider a balls-into-bins process with $s$ balls and $t$ bins. Throw each ball to a bin sampled uniformly at random from the set of all bins. Let $Z$ be the number of non-empty bins at the end of this process. Then, it is straightforward to observe that the expected number of non-empty bins is
$$\E[Z] = t \left(1-\left(1-\frac{1}{t}\right)^s\right) \ge t \left(1-e^{-s/t}\right) .$$
By Lemma~1 in \cite{kane10:_optim}, for $100 \le s \le t/2$, the variance of the number of non-empty bins satisfies\footnote{The constants used here are slightly different from~\cite{kane10:_optim}.}
$$
\var[Z] \le 5\frac{s^2}{t}
$$

Let $R'$ be the number of non-empty bins in the balls-into-bins process with $s = 2k/5$ balls and $t = 4k/5$ bins. Then, we have
$$
\E[R'] \ge \frac{4}{5}k \left(1-1/\sqrt{e}\right)
$$
and
$$
\var[R'] \le 5\frac{(2k/5)^2}{4k/5} = k. 
$$

By the Chebyshev's inequality,
$$
\Pr[R' < \E[R'] - k/20] \le \frac{\var[R']}{(k/20)^2} < \frac{320}{k}.
$$
Hence, with probability $1 - O(1/k)$, $R \ge R' \ge k/4$, which proves the claim in (\ref{cla:binball}).

It follows that for each $G^j$, we can find a matching in $G^j$ of size at least $k/4$ with probability $1 - O(1/k)$. If $r = n/(2k) = o(k)$, then by the union bound, it holds that with probability at least $1-1/100$, the size of a maximum matching in $G$ is at least $n/4$. Otherwise, let $R^1, R^2, \ldots, R^r$ be the sizes of matchings that are independently computed using the greedy matching algorithm described above for respective input graphs $G^1,G^2, \ldots, G^r$. Let $Z_j = 1$ if $R^j \ge k/4$, and $Z_j = 0$, otherwise. Since $R_j \ge k Z_j/4$ for all $j\in [r]$ and $\E[Z_j] = 1 - O(1/k)$, by the Hoeffding's inequality, we have
$$
\Pr\left[\sum_{j=1}^r R^j < \frac{n}{5}\right] \le \Pr\left[\sum_{j=1}^r Z_j < \frac{4n}{5k}\right] \le e^{-\Omega(r)}
$$
Hence, the size of a maximum matching in $G$ is at least $n/5$ with probability at least $1-e^{-\Omega(r)} \ge 1-1/100$.
\end{proof}

If $\P'$ is an $\alpha$-approximation algorithm with error probability at most $1/4$, then by Lemma~\ref{lem:match-property}, with probability at least $3/4 - 1/100 \ge 2/3$, $\P'$ will output a matching $M$ that contains at least $\alpha n/5 - 2pn$ important edges, and we denote this event by $\mathcal{F}$. We know that there are at most $n/2$ important edges and edge $(u^{J, I}, v^{J, L})$ is one of them. We say that $(i,j,l)$ is important for $G$, if $(u^{j,i}, v^{j,l})$ is an important edge in $G$. Given an input $G$, the algorithm cannot distinguish between any two important edges. We can apply the principle of deferred decisions to decide the value of $(I,J)$ after the matching has already been computed, which means, conditioned on $\mathcal{F}$, the probability that $(u^{J, I}, v^{J, L}) \in M$ is at least $(\alpha n/5 - 2pn)/(n/2) = \alpha/5$, where $p = \alpha/20$. Since $\mathcal{F}$ happens with probability at least $2/3$, we have 
$$
\Pr[(u^{J, I}, v^{J, L}) \in M]\ge \alpha/10.
$$

To sum up, we have shown that protocol $\P$ solves \DISJ\ correctly with one-sided error of at most $1 - \alpha/10$.

\paragraph{Information cost} We analyze the information cost of \match. Let $\Pi = \Pi^1 \circ \Pi^2 \circ \cdots \circ \Pi^k$ be the best protocol for \match\ with respect to input distribution $\phi$ and the one-sided error probability $1 - \alpha/10$. 

Let $W^{-J} = (W^1, \ldots,W_{J-1},W_{J+1},\ldots, W^r)$, and $W = (W^1,W^2,\ldots,W^r)$. Let $W_{A,B} \sim \tau_p^k$ denote the random variable used to sample $(A,B)$ from $\mu_k$. Recall that in our input reduction $I, J, W^{-J}$ are public coins used by Alice and Bob. 

We have the following:
\begin{eqnarray}
 \frac{2}{n} IC_{\phi, \delta}(\hbox{\match}) 
& \ge & \frac{1}{rk} \sum_{i=1}^k I(X^i, Y; \Pi^i) \quad  \label{ffineq} \\
& \ge &   \frac{1}{rk} \sum_{i = 1}^k I(X^i, Y; \Pi^i | W)\label{fineq} \\
& \ge &  \frac{1}{rk} \sum_{i = 1}^k \sum_{j=1}^{r} I(X^{j, i}, Y^j; \Pi^i | W^{-j}, W^{j})\label{eq:b-0} \\
& = &  \frac{1}{rk} \sum_{i = 1}^k \sum_{j=1}^{r} I(A, B; \Pi^i | I = i, J = j, W^{-j}, W_{A,B})\label{eq:b-1} \\
& = &  I(A, B; \Pi^I | I , J , W^{-J}, W_{A,B} ) \nonumber \\
& \ge &  I(A, B; \Pi^* | W_{A,B}, R)  \label{eq:b-2}   \\
& = & \Omega(\alpha^2 k), \label{eq:b-3}
\end{eqnarray}
where (\ref{ffineq}) is by Lemma~\ref{lem:IC-in-MP}, (\ref{fineq}) is by data processing inequality, (\ref{eq:b-0}) is by the super-additivity property, (\ref{eq:b-1}) holds because the distribution of $W^j$ is the same as that of $W_{A,B}$, and the conditional distribution of $(X^{j, i}, Y^j, \Pi^i)$ given $W^{-j}, W^{j}$ is the same as the conditional distribution of $(A,B, \Pi^i)$ given $I = i$, $J = j$, $W^{-j}$, $W_{A,B}$, in (\ref{eq:b-2}), $\Pi^*$ is the best protocol for \DISJ\ with one-sided error probability at most $1 - \alpha/10$ and $R$ is the public randomness used in $\Pi^*$, and (\ref{eq:b-3}) holds by Theorem~\ref{thm:DISJ} where recall that we have set $p = \alpha/20$.

We have thus shown that $IC_{\phi, 1/4}(\hbox{\match}) \ge \Omega(\alpha^2 kn)$. Since by Theorem~\ref{thm:R-IC}, $R_{1/4}(\hbox{\match}) \ge IC_{\phi, 1/4}(\hbox{\match})$, it follows that
$$
R_{1/4}(\hbox{\match}) \ge \Omega(\alpha^2 kn)
$$
which proves Theorem~\ref{thm:main}.

\section{Upper Bound}\label{sec:up}

In this section we present an $\alpha$-approximation algorithm for distributed matching problem with an upper bound on the communication complexity that matches the lower bound for any $\alpha\le 1/2$ up to poly-logarithmic factors. 

We have described a simple algorithm that guarantees an $1/2$-approximation for \match\ at the communication cost of $O(kn \log n)$ bits in Section~\ref{sec:intro}. This algorithm is a greedy algorithm that computes a maximal matching. The communication cost of the algorithm is $O(\alpha^2 kn\log n)$ bits. If $1/8< \alpha\le 1/2$, we simply apply the greedy $1/2$-approximation algorithm that has the communication cost of $O(kn\log n)$ bits. Therefore, we assume that $ \alpha \le 1/8$ in the rest of this section. We next present an $\alpha$-approximation algorithm that uses the greedy maximal matching algorithm as a subroutine. 

\paragraph{} {\bf Algorithm:} The algorithm consists of two steps: 

\begin{enumerate}
\item The coordinator sends a message to each site asking to compute a local maximum matching, and each site then follows up with reporting back to the coordinator the \emph{size} of its local maximum matching. The coordinator sends a message to a site that holds a local maximum matching of maximum size, and this site then responds with sending back to the coordinator at most $\alpha n$ edges from its local maximum matching. Then, the algorithm proceeds to the second step.  
 
\item The coordinator selects each site independently with probability $q$, where $q$ is set to $8\alpha$ (recall we assume $\alpha\le 1/8$), and computes a maximal matching by applying the greedy maximal matching algorithm to the selected sites.
\end{enumerate}

It is readily observed that the expected communication cost of Step~1 is at most $O((k+\alpha n) \log n)$ bits, and that the communication cost of Step~2 is at most $O((k+\alpha^2 kn) \log n)$ bits. We next show correctness of the algorithm.

\paragraph{Correctness of the algorithm.} Let $X_i$ be a random variable that indicates whether or not site $p^i$ is selected in Step~2. Note that $\E[X_i] = q$ and $\var[X_i] = q(1 - q)$. Let $M$ be a maximum matching in $G$ and let $m$ denote its size. Let $m_i$ be the number of edges in $M$ which belong to site $p^i$. Hence, we have $\sum_{i=1}^k m_i = m$ because the edges of $G$ are assumed to be partitioned disjointly over the $k$ sites. We can assume that $m_i \le \alpha m$ for all $i \in [k]$; otherwise, the coordinator has already gotten an $\alpha$-approximation from Step~1.

Let $Y$ be the size of the maximal matching that is output of Step~2. Recall that any maximal matching is at least $1/2$ of any maximum matching. Thus, we have $Y \ge X/2$, where $X = \sum_{i=1}^k m_i X_i$. Note that we have $\E[X] = q m$ and $\var[X] = q (1 - q)\sum_{i=1}^k m_i^2$. Under the constraint $m_i \le \alpha m$ for all $i\in [k]$, we have 
$$\sum_{i=1}^k m_i^2 \le \alpha m\sum_{i=1}^km_i =\alpha m^2.$$
Hence, combining with the assumption $q = 8\alpha$, it follows that $\var[X] \le 8\alpha^2m^2$. By Chebyshev's inequality, we have
$$
\Pr[|X - q m|\ge 6\alpha m] \le \frac{8}{36} < \frac{1}{4}.
$$
Since $q=8\alpha$, it follows that $X \ge 2\alpha m$ with probability at least $3/4$. Combining with $Y\ge X/2$, we have that $Y \ge \alpha m$ with probability at least $3/4$.

We have shown the following theorem. 

\begin{theorem} For every $\alpha\le 1/2$, there exists a randomized algorithm that computes an $\alpha$-approximation of a maximum matching with probability at least $3/4$ at the communication cost of $O((\alpha^2 kn + \alpha n + k)\log n)$ bits.
\end{theorem}

Note that $\Omega(\alpha n)$ is a trivial lower bound, simply because the size of the output could be as large as $\Omega(\alpha n)$. Obviously, $\Omega(k)$ is a lower bound, because the coordinator has to send at least one message to each site. Thus, together with the lower bound $\Omega(\alpha^2 kn)$ in Theorem~\ref{thm:main}, the upper bound above is tight up to a $\log n$ factor.

One can see that the above algorithm needs $O(\alpha k)$ rounds, as we use a naive algorithm to compute a maximal matching among $\alpha k$ sites. If $k$ is large, say, $n^\beta$ for some constant $\beta\in(0,1)$, this may not be acceptable. Fortunately, Luby's parallel algorithm \cite{luby1986simple} can be easily adapted to our model, using only $O(\log n)$ rounds at the cost of increasing the communication by at most a $\log n$ factor. The details are provided in Appendix~\ref{app:luby's}. 

\section{Conclusion}\label{sec:conc}
We have established a tight lower bound on the communication complexity for approximate maximum matching problem in the message-passing model. 

An interesting open problem is the complexity of the counting version of the problem, i.e., the communication complexity if we only want to compute an approximation of the \emph{size} of a maximum matching in a graph. Note that our proof of the lower bound relies on the fact that the algorithm has to return a certificate of the matching. Hence, in order to prove a lower bound for the counting version of the problem, one may need to use new ideas and it is also possible that a better upper bound exists. In a recent work \cite{KKS14}, the counting version of the matching problem was studied in the random-order streaming model. They proposed an algorithm that uses one pass and poly-logarithmic space, which computes a poly-logarithmic approximation of the size of a maximum matching in the input graph. 

A general interesting direction for future research is to investigate the communication complexity for other combinatorial problems on graphs, for example, connected components, minimum spanning tree, vertex cover and dominating set. The techniques used for approximate maximum matching in the present paper could be of use in addressing these other problems.

\appendix
\section{Luby's algorithm in the coordinator model}\label{app:luby's}
{\bf Luby's algorithm~\cite{luby1986simple}:} Let $G=(V,E)$ be the input graph, and $M$ be a matching initialized to $\emptyset$. Luby's algorithm for maximal matching is as follows.
\begin{enumerate}
\item If $E$ is empty, return $M$.\label{algostep:luby}
\item Randomly assign unique priority $\pi_e$ to each $e\in E$.
\item Let $M'$ be the set of edges in $E$ with higher priority than all of its neighboring edges. Delete $M'$ and all the neighboring edges of $M'$ from $E$, and add $M'$ to $M$. Go to step~\ref{algostep:luby}.
\end{enumerate}
It is easy to verify that the output $M$ is a maximal matching. The number of iterations before $E$ becomes empty is at most $O(\log n)$ in expectation~\cite{luby1986simple}. Next we briefly describe how to implement this algorithm in the coordinator model. Let $E^i$ be the edges held by site $p^i$.
\begin{enumerate}
\item For each $i$, if $E^i$ is empty, $p^i$ halts. Otherwise $p^i$ randomly assigns unique priority $\pi_e$ to each $e\in E^i$.\label{algostep:lubyMP}
\item Let $M'^i$ be the set of edges in $E^i$ with higher priority than all of its neighboring edges in $E^i$. Then $p^i$ sends $M'^i$ together with their priorities to the coordinator.
\item Coordinator gets $W=M'^1\cup M^2 \cup \cdots\cup M'^k$. Let $M'$ be the set of edges in $W$ with higher priority than all of its neighboring edges in $W$. Coordinator adds $M'$ to $M$ and then sends $M'$ to all sites.
\item Each site $p^{i}$ deletes all neighboring edges of $M'$ from $E^{i}$, and goes to step~\ref{algostep:lubyMP}.
\item After all the sites halt, the coordinator outputs $M$.
\end{enumerate}

It is easy to see that the above algorithm simulates the algorithm of Luby. Therefore, the correctness follows from the correctness of Luby's algorithm, and the number of rounds is the same, which is $O(\log n)$. The communication cost in each round is at most $O(kn\log n)$ bits because, in each round, each site sends a matching to the coordinator, and the coordinator sends back another matching. Hence, the total communication cost is $O(kn\log^2n)$ bits.

%\begin{acknowledgements}
%\end{acknowledgements}

%\bibliography{ref}   % name your BibTeX data base

\end{document}